\documentclass[11pt]{article}
\usepackage{fullpage}
\usepackage[utf8]{inputenc}
\usepackage{listings}
\usepackage{amsmath}
\usepackage{comment}
\usepackage[skins, breakable]{tcolorbox}
\usepackage{multirow}
\usepackage{amsthm}
\usepackage{amsfonts}
\usepackage{graphicx}
\usepackage{appendix}
\usepackage{xspace}
\usepackage{enumitem}
\usepackage{bm}
\usepackage{floatflt}
\usepackage{thmtools}
\usepackage{thm-restate}

\newtheorem{lemma}{Lemma}
\newtheorem{theorem}{Theorem}
\newtheorem{definition}{Definition}

\newcommand{\fb} {\mathcal{B}}
\newcommand{\fc} {\mathcal{C}}

\newcommand{\fs} {\mathcal{S}}

\newcommand{\latency} {good-case latency\xspace}

\newcommand{\ack}{\texttt{ack}\xspace}
\newcommand{\ackone}{\texttt{ack}\xspace}

\newcommand{\vote}{\texttt{vote}\xspace}
\newcommand{\voteone}{\texttt{vote1}\xspace}
\newcommand{\votetwo}{\texttt{vote2}\xspace}

\newcommand{\blame}{\texttt{blame}\xspace}
\newcommand{\blameone}{\texttt{blame1}\xspace}
\newcommand{\blametwo}{\texttt{blame2}\xspace}

\newcommand{\status}{\texttt{status}\xspace}

\newcommand{\clock}{\sigma\xspace}

\newcommand{\bb}{$1\Delta$-BB\xspace}
\newcommand{\smr}{$1\Delta$-SMR\xspace}

\newcommand{\smrms}{$1\Delta$-SMR-MSF\xspace}

\newcommand{\ba}{$1\Delta$-BA\xspace}

\newcommand{\locked}{b_{lck}\xspace}

\newcommand{\stitle}[1]{\vspace{0.8ex} \noindent\textsf{\textbf{#1}}}

\renewcommand{\paragraph}[1]{\smallskip\stitle{#1}}

\newtcolorbox{mybox}[1][]{
enhanced,
colback=white,
boxsep=0pt,
#1
}

\newlist{numlist}{enumerate}{10}
\setlist[numlist]{label*=\arabic*.}

\usepackage{xpatch}
\xpatchcmd{\declaretheorem}{{thmdef}{#1}}{{thmdef}{numberlike=theorem,#1}}{}{}

\usepackage{authblk}

\begin{document}

\title{Byzantine Agreement, Broadcast and State Machine Replication with Near-optimal Good-case Latency~\thanks{Kartik Nayak and Ling Ren are supported in part by a VMware early career faculty grant.}} 

\author[1]{Ittai Abraham}
\author[2]{Kartik Nayak}
\author[3]{Ling Ren}
\author[4]{Zhuolun Xiang~\thanks{Lead author}}
\affil[1]{VMware Research\\ {iabraham@vmware.com}}
\affil[2]{Duke University\\ {kartik@cs.duke.edu}}
\affil[3, 4]{University of Illinois at Urbana-Champaign\\ {\{renling, xiangzl\}@illinois.edu}}

\maketitle

\begin{abstract}
This paper investigates the problem  \textit{good-case latency} of Byzantine agreement,  broadcast and state machine replication in the synchronous authenticated setting. The good-case latency measure captures the time it takes to reach agreement when all non-faulty parties have the same input (or in BB/SMR when the sender/leader is non-faulty).
Previous result implies a lower bound showing that any Byzantine agreement or broadcast protocol tolerating more than $n/3$ faults must have a good-case latency of at least $\Delta$ \cite{synchotstuff}, where $\Delta$ is the assumed maximum message delay bound. 
Our first result is a  family of protocols we call $1\Delta$ that have near-optimal \latency. We propose a protocol $1\Delta$-BA that solves Byzantine agreement in the synchronous and authenticated setting with near-optimal good-case latency of $\Delta+2\delta$ and optimal resilience $f<n/2$, where $\delta$ is the actual (unknown) delay bound. 
We then extend our protocol and present $1\Delta$-BB and $1\Delta$-SMR for Byzantine fault tolerant broadcast and state machine replication, respectively, in the same setting and with the same good-case latency of $\Delta+2\delta$ and $f<n/2$ fault tolerance.
Our $1\Delta$-SMR upper bound improves the gap between the best current solution, Sync HotStuff, which obtains a good-case latency of $2\Delta$ per command and the lower bound of $\Delta$ on good-case latency.
Finally, we investigate weaker notions of the synchronous setting and show how to adopt the $1\Delta$ approach to these models.
\end{abstract}

\section{Introduction}

\label{sec:intro}

Byzantine agreement (BA) and Byzantine broadcast (BB) are fundamental problems in distributed computing. 
In {\em Byzantine agreement}, each replica has an input and each must decide on an output. All non-faulty (honest) replicas must decide on the same output. Moreover, if all honest replicas have the same input then this must be the decision output.
In \textit{Byzantine broadcast}, there is a {\em designated sender} that sends its input to all replicas, and all honest replicas must decide on the same output.
Moreover, if the sender is honest, then the decision output must be sender's input.
One of the most important practical applications of BA and BB is to implement {\em Byzantine fault tolerant (BFT) state machine replication (SMR)} \cite{schneider1990implementing}, which provides clients with the illusion of a single honest server, 
by ensuring that all honest replicas agree on the same sequence of client inputs.


It is well known that under asynchrony or partial synchrony, the optimal resilience for any BA, BB or BFT SMR protocol is $f<n/3$ \cite{dwork1988consensus, castro1999practical}.
In order to move beyond $f<n/3$, we need to assume a synchronous \cite{dwork1988consensus} and authenticated \cite{fischer1990easy} setting. The optimal resilience for BA and BFT SMR in the synchronous authenticated setting is $f<n/2$, and $f<n$ for BB in the same setting.
However, synchronous BA, BB and BFT SMR have been considered impractical for a long time for several reasons.


Firstly, to take advantage of the synchrony assumption, most theoretical works on Byzantine agreement and broadcast assume lock-step execution~\cite{LSP82,dolev1983authenticated,dolev1990early,katz2006expected}, where replicas start and end each round at the same time.
As a result, the latency of BA/BB protocols is typically measured by the round complexity.
Such a lock-step assumption simplifies the protocol design but it is considered impractical because it not only is hard to enforce but also leads to poor performance. 
A synchronous protocol requires a known upper bound $\Delta$ for the maximum message delay. 
To be safe under worst-case network conditions, $\Delta$ has to be picked conservatively, i.e., much larger than the actual (unknown) message delay bound.
If a protocol that runs in lock steps, it must allocate $\Delta$ time for every step.
Then, most of the time is wasted on waiting for the next round to start rather than doing any useful work.  
Only recently, synchronous BFT SMR protocols deviated from the lock-step approach \cite{hanke2018dfinity,synchotstuff}.
In these protocols, most of the steps are non-blocking, namely, replicas move to the next step as soon as enough messages are received from the previous step. 

Secondly, classical BA/BB protocols \cite{dolev1983authenticated} tend to optimize their \emph{worst-case} latency.
Because the worst-case number of rounds required is $f+1$ for tolerating $f$ faults \cite{fischer1982lower} and $f$ is typically assumed to be linear in $n$, any BA/BB protocol will inevitably have a poor worst-case latency as $n$ increases.
However, in contrast, BFT SMR protocols typically care about the \emph{good-case}, in which a stable honest leader stays in charge and drives consensus on many decisions. 
While a line of work studies \emph{expected-case} \cite{feldman1988optimal,katz2006expected,abraham2019synchronous} latency, this metric is still very different from the \latency in SMR, because it essentially analyzes the expected number of times the protocol changes its leader.

Another relatively minor disconnect lies in the ``life cycle'' of the protocol.
Theoretical BA/BB considers consensus on a single value and require all replicas to \emph{halt} or \emph{terminate} after agreeing on this single value.
In contrast, practical SMR protocols are intended to run forever; replicas \emph{commit} or \emph{decide} on an ever-growing sequence of values.

Motivated by the above considerations, we argue that 
there is a need for a more refined 
model for BA and BB 
%
to better capture practice. 
First, lock-step execution should not be assumed
.
Latency
should be measured in time as opposed to rounds.
Furthermore, for a more accurate characterization of latency, we adopt the separation between the conservative bound $\Delta$ and the actual (unknown) bound $\delta$ as suggested in \cite{ierzberg1989efficient, pass2017hybrid}.
Finally, instead of measuring only the traditional \emph{worst-case} latency to terminate, we use the \latency to commit as the main metric, which is defined as follows. 
\begin{definition}[\latency for Byzantine fault tolerant state machine replication]\label{def:goodcase}
    The \latency is the maximal latency (over all adversarial strategies) until all honest replicas commit given an honest leader is in charge.
\end{definition}

From the \latency definition of BFT SMR above, it is natural to define the {\em \latency for BB} by replacing the good leader property  with \emph{good sender}: the designated sender is honest.
Similarly, we define the {\em \latency for BA} by replacing the good leader property  with \emph{good input}: all honest replicas have the same input. 
The good input property also ties back to the good-case of BFT SMR because if the leader is honest, then all honest replicas receive the same input from the leader.

The main goal of our paper is to develop BA and BB protocols with near-optimal \latency and apply them to BFT SMR protocols under the synchronous and authenticated setting.
First, noticing that for $f<n/3$, one can use any partially synchronous or asynchronous protocol (possibly in unauthenticated setting) to achieve a \latency that only consists $\delta$, since these protocols can proceed in network speed under this setting.
As mentioned, in order to move beyond $f<n/3$ and to obtain the optimal resilience $f<n/2$ for BA and BFT SMR, we need to assume a synchronous \cite{dwork1988consensus} and authenticated \cite{fischer1990easy} setting.
We first restate a lower bound of $\Delta$ on the \latency from Sync HotStuff \cite{synchotstuff} which closely follows \cite{dwork1988consensus}. 

\begin{restatable}{theorem}{lowerbound}
\label{thm:lb}
Any Byzantine agreement or broadcast protocol that is resilient to $f \geq n/3$ faults must have a \latency at least $\Delta$.
\end{restatable}

It is worth mentioning that a lower bound in the earlier work~\cite{attiya1991bounds} implies a $\Delta$ lower bound on \latency for crash-tolerant agreement with $f\geq n/2$ faults, but it does not directly implies the lower bound above with $f\geq n/3$ Byzantine faults.



The state of the art \latency of BA/BB protocols in the synchronous authenticated setting is implied by Sync HotStuff \cite{synchotstuff}. Sync HotStuff is a BFT SMR protocol with \latency of $2\Delta$, and a single-shot version of Sync HotStuff naturally implies a BA/BB protocol with \latency of $2\Delta$. In fact, the authors in \cite{synchotstuff} conjectured that $2\Delta$ is the optimal \latency possible. 
In this paper, we refute this conjecture and present the following main results.

\paragraph{Byzantine Agreement, Broadcast, and State Machine Replication with Near-optimal Good-case Latency of $\Delta+2\delta$.}
A key contribution of this paper is the first Byzantine agreement protocol \ba in the synchronous authenticated setting with near-optimal \latency of $\Delta+2\delta$ and optimal resilience of $f<n/2$. 
We also obtain a Byzantine broadcast protocol \bb with near-optimal \latency of $\Delta+2\delta$ after small modifications.
By proposing protocols \ba and \bb with \latency of $\Delta+2\delta$, we improve the gap between the upper bound and lower bound on \latency for Byzantine agreement and broadcast.
We then give a state machine replication protocol named \smr that has \latency of $\Delta+2\delta$ and optimal resilience of $f<n/2$.

\paragraph{Extensions with near-optimal \latency under weaker models.}
To make the synchronous model more practical, we consider two types of additional fault suggested in the literature  named mobile link failures and mobile sluggish faults.
The mobile link failure model \cite{schmid2009impossibility, biely2011synchronous} assumes a certain fraction of send and receive links can be down at each replica. 
The mobile sluggish fault model \cite{chan2018pili, guo2019synchronous, synchotstuff} considers the slow connection of some honest replicas whose message sending and receiving does not respect the assumed delay bound $\Delta$.
We show that our protocols can be extended to tolerate the mobile link failures and mobile sluggish faults.
The \latency becomes $2\Delta+4\delta$ under mobile link failures, and we prove a new lower bound of $2\Delta$.
With mobile sluggish faults, the \latency becomes $\Delta+4\delta$.

\section{Preliminary}
\label{sec:pre}

We consider $n$ replicas in a reliable, authenticated all-to-all network, where up to $f$ replicas can be malicious and behave in a Byzantine fashion, and rest of the replicas are honest.
We assume standard digital signatures and public-key infrastructure (PKI). 
We use $\langle x \rangle_p$ to denote a signed message $x$ by replica $p$.

\paragraph{Pessimistic message delay bound $\Delta$ and actual message delay bound $\delta$.}
In this paper, we consider a synchronous system, where the message delay is bounded.
Let $\delta$ denote the actual upper bound of the message delay in the network, so any message from a non-faulty sender will be delivered within $\delta$ time after being sent. 
The parameter $\delta$ is unknown to the protocol designer and used only in the model definition, not in any protocol.
The synchronous model assumes a known upper bound $\Delta$ for the message delay $\delta$, i.e., $\delta\leq \Delta$. 
In practice, the parameter $\Delta$ is usually conservative for safety reasons, and thus $\delta\ll\Delta$.
To put the new model into context, we remark on its natural connection to partial synchrony: in partial synchrony, there is also an unknown message delay bound $\delta$ but the protocol designer does not know any upper bound on $\delta$.

\paragraph{Non-lock-step and clock skew.}
Although we assume synchrony, unlike most synchronous protocols that require lock-step execution (replicas start and end each round at the same time), our protocols do not proceed in a lock-step fashion, following the recent synchronous BFT SMR protocols which deviated from the lock-step approach \cite{hanke2018dfinity,synchotstuff}.
We assume replicas have clock skew at most $\clock$, i.e., they start protocol at most $\clock$ apart from each other. 
The bounded clock skew can be guaranteed via any clock synchronization protocol, such as \cite{dolev1995dynamic, abraham2019synchronous}.
The above clock synchronization protocols can also handle bounded clock drifts.
For simplicity, we assume there's no clock drift.


\begin{definition}[Byzantine Agreement]
\label{def:ba}
A Byzantine agreement protocol provides the following three
guarantees.
    \begin{itemize}[noitemsep,topsep=0pt]
        \item Agreement. If two honest replicas commit value $b$ and $b'$ respectively, then $b=b'$.
        \item Termination. All honest replicas eventually commit and terminate.
        \item Validity. 
        If all honest replicas have the same input value, then all honest replicas commit on the value.
    \end{itemize}
\end{definition}

The {\em Byzantine broadcast} is defined by changing the validity to require that if the designated sender is honest, then all honest replicas commit on the sender's value.


\begin{definition}[Byzantine Fault tolerant State Machine Replication \cite{synchotstuff}]
\label{def:smr}
A Byzantine fault tolerant state machine replication protocol commits client requests as a linearizable log to provide a consistent view of the log akin to a single non-faulty server, providing the following two guarantees.
    \begin{itemize}[noitemsep,topsep=0pt]
        \item Safety. Honest replicas do not commit different values at the same log position. 
        \item Liveness. Each client request is eventually committed by all honest replicas.
    \end{itemize}
\end{definition}

\section{Byzantine Agreement/Broadcast with Near-optimal Good-case Latency}
\label{sec:ba}

We first present a synchronous Byzantine agreement protocol \ba with near-optimal \latency of $\Delta+2\delta$ and optimal resilience of $f<n/2$. 
Without loss of generality, we assume  $n=2f+1$.
Our protocol incurs a \latency of $\Delta+2\delta$, which is near-optimal due to the $\Delta$ lower bound (Theorem \ref{thm:lb}).
The protocol can be naturally extended to a protocol \bb that solves Byzantine broadcast with near-optimal \latency.
Our protocols improves the gap for the \latency of synchronous BA and BB and refutes a conjectured $2\Delta$ lower bound in Abraham et al.~\cite{synchotstuff}.

\subsection{Intuition}

\begin{floatingfigure}[r]{0.4\linewidth}
	\includegraphics[width=0.4\linewidth]{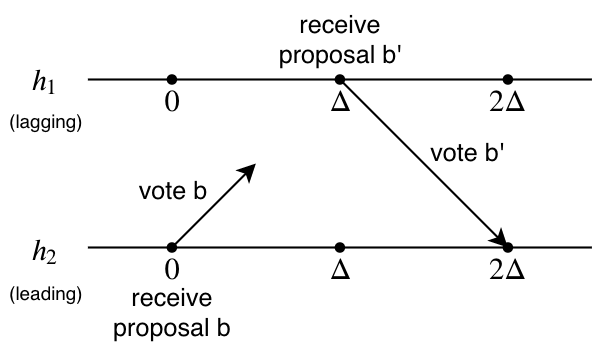}
	\caption{Graphical representation of the intuition for the conjectured $2\Delta$ lower bound.}
\label{fig:lower-bound}
\end{floatingfigure}

We start by presenting the rationale of the conjectured $2\Delta$ lower bound~\cite{synchotstuff} and how our protocol disproves it. 
We present the intuition in the context of Byzantine broadcast for simplicity, but similar arguments hold for Byzantine agreement as well.
The conjecture argues that two $\Delta$ periods are needed for the following reasons. 
First, now that we have departed from lock-step execution, two honest replicas may be ``out-of-sync'' by $\Delta$ time.
For example, a Byzantine leader quickly delivers its first message to one replica but takes $\Delta$ time to deliver its first message to another replica. 
Intuitively, this causes a \emph{lagging} replica to reach a point in the protocol $\Delta$ time after a \emph{leading} replica. 
Second, any message sent by an honest replica can take up to a $\Delta$ time to arrive at another replica by the synchrony assumption. 
In order to make sure that the sender did not equivocate (i.e., send different values to different replicas), a leading replica $h_2$ needs to wait for $2\Delta$ time before it can hear from a lagging replica $h_1$ what $h_1$ received from the sender (cf. Figure~\ref{fig:lower-bound}). 
Lastly, replicas seem to have no way to tell whether they are leading or lagging, so all replicas wait for $2\Delta$ time to ensure the absence of a sender equivocation. 


Note that the above intuition of the conjecture has an implicit assumption: if the sender equivocates, we want all replicas, leading or lagging, to detect sender's equivocation.
This is where the intuition of the conjecture errs and what our protocol relies on to get below $2\Delta$: we will make equivocation detection \emph{asymmetric}.
To elaborate, if we reduce the waiting period to $1\Delta$, a lagging replica $h_1$ can still learn from a leading replica $h_2$ about what $h_2$ receives from the sender, but not vice versa. 
In other words, a lagging replica can still detect sender equivocation (if there's any) but the most leading replica may not.
Hence, the most leading replica may commit a value despite that the sender equivocated.
But note that, all other honest replicas have detected equivocation and hence, have not committed.
As long as we carefully craft the rest of the protocol to make all other honest replicas eventually commit the same value as the leading replica, the protocol is safe.
Also note that this is the Byzantine-sender scenario, so the rest of the protocol does not have to meet the $1\Delta$ deadline.

\medskip
\subsection{\ba Protocol under Synchrony}
\label{sec:ba:protocol}

We now describe the \ba presented in Figure~\ref{fig:ba:sync}, which starts at Step~\ref{ba:step:exchange} and proceeds in an event-driven manner where a step is executed if certain conditions are satisfied.
For simplicity, we first assume that all replicas start the protocol at the same time $t=0$, and our results can be easily extended to the case where there is clock skew (see the discussion later in this section). 
Initially, all replicas have an input, and set their locked value $b_{lck}$ to be some default value $\bot$.
When protocol starts, all replicas first sign and broadcast its input value, and try to form a proposal containing $f+1$ signed messages of an identical value (Step~\ref{ba:step:exchange}). 
When any replica $i$ forms a proposal of value $b$ from the previous step or receives a valid proposal of value $b$ from other replicas, it forwards the proposal to detect conflict (Step~\ref{ba:step:forward}).
%
Due to the non-lock-step execution of our protocol, it is possible that a replica receives the proposal forwarded by some other replica first, then forms its own proposal.
Step~\ref{ba:step:forward} may be executed multiple times, if the replica observes multiple valid proposals.
If no valid proposal is ever formed or received, all replicas do nothing until time $4\Delta$ and invoke a BA at Step~\ref{ba:step:ba} with input $b_{lck}=\bot$, which ensures the agreement and termination of the protocol.
Otherwise, after the replica forwards the proposal, it locally starts a timer called $\texttt{vote-timer}$ to wait for $\Delta$ time period. 
If the timer expires and the replica does not receive any conflicting proposal containing a different value $b'\neq b$, it broadcasts a \vote message for the proposal (Step~\ref{ba:step:vote}). 
Once the replica gathers $f+1$ distinct \vote messages for the proposal containing value $b$ within time $3\Delta$, it broadcasts these \vote messages, sets locked value $b_{lck}=b$ and commits the value $b$. 
If the $f+1$ distinct \vote messages for value $b$ are received later than $3\Delta$, the replica only sets locked value $b_{lck}=b$ without committing the value (Step~\ref{ba:step:commit}).
As will be proved in Theorem \ref{thm:ba}, in the good-case, every honest replica will receive $f+1$ \vote messages within $3\Delta$ time. 
For other cases, if some replica commits to a value by time $3\Delta$, since it forwards the $f+1$ \vote messages to all replicas, all replicas will set $b_{lck}$ to the committed value within $4\Delta$ time.
%
At time $4\Delta$, the replica initiates an instance of Byzantine agreement with its locked value as the input, and any replica that has not committed yet will commit on the output of the agreement (Step~\ref{ba:step:ba}).

\begin{figure}[tb]
    \centering
    \begin{mybox}
    Initially, every replica $i$ has an input $b_i$, starts the protocol at the same time $t=0$, and sets $\locked=\bot$.
\begin{enumerate}
    \item\label{ba:step:exchange} \textbf{Propose.} Sign and send the input value $b_i$ to all others. 
    Once receiving $f+1$ distinct signed messages of the same value $b$, form a proposal $B$ with these messages as $B=\{\langle b\rangle_j\}_{f+1}$.

    \item\label{ba:step:forward} \textbf{Forward.} 
    Upon forming or receiving a valid new proposal $B$ containing $f+1$ distinct signed messages of the same value $b$, 
    forward $B$ to all other replicas, 
    set $\texttt{vote-timer}_B$ for proposal $B$ to $\Delta$ and start counting down.
    
    \item\label{ba:step:vote} \textbf{Vote.} 
    When $\texttt{vote-timer}_B$ for proposal $B$ containing value $b$ reaches $0$,
    if the replica does not receive another valid proposal $B'$ containing $f+1$ distinct signed messages of a different value $b'\neq b$, 
    it broadcasts a vote in the form of $\langle \texttt{vote}, B \rangle_i$.
    
    \item\label{ba:step:commit} \textbf{Commit.}
    Upon collecting $f+1$ distinct signed votes $\langle \texttt{vote}, B \rangle$ of a valid proposal $B$ containing value $b$ at time $t$, 
    (i) if $t\leq 3\Delta$, it broadcasts these $f+1$ votes, sets $\locked=b$, and commits $b$, 
    (ii) if $t> 3\Delta$, it sets $\locked=b$.
    
    
    
    \item\label{ba:step:ba} \textbf{Byzantine agreement.}
    At time $4\Delta$, invoke an instance of Byzantine agreement with $\locked$ as the input. 
    If not committed, commit on the output of the Byzantine agreement. 
    Terminate.

\end{enumerate}
    \end{mybox}
    \caption{\ba Protocol under the Synchronous Model}
    \vspace{-1em}
    \label{fig:ba:sync}
\end{figure}


\paragraph{Why does it suffice to wait for $\Delta$ time before sending \vote?} Consider any two honest replicas $h$ and $h'$ who receive conflicting proposals containing values $b$ at time $t$ and $b'$ at time $t'$ respectively. Without loss of generality, suppose $t \leq t'$. Observe that $h$'s forwarded proposal will arrive at $h'$ by $t+\Delta \leq t'+\Delta$, which is within the $\Delta$ waiting period of $h'$. Thus, $h'$ will not vote and will not commit.
This argument holds for any two pairs of honest replicas, and hence only the value of the proposal first voted by any honest replica may potentially be committed by a set of honest replicas at Step~\ref{ba:step:commit}. 
Now these committed honest replicas will forwards the $f+1$ distinct votes to all replicas no later than time $3\Delta$. 
All other honest replicas will receive these by time $4\Delta$ and will set $b_{lck}$ to the committed value. Thus, in the Byzantine agreement at Step~\ref{ba:step:ba}, all honest replicas will start with the same value, and the output of BA will be the committed value to ensure agreement among all honest replicas.

\paragraph{Remark on the Byzantine agreement primitive.}
We can plug in any Byzantine agreement protocol that tolerates $f<n/2$ faults and satisfies standard Byzantine agreement definition (Definition \ref{def:ba}) at Step~\ref{ba:step:ba}.
Note that the latency of the BA protocol does not affect the \latency of \ba protocol, since in the good-case, then all honest replicas can commit the value at Step \ref{ba:step:commit} of the protocol before invoking the BA (Theorem \ref{thm:ba}). 
For the same reason, it is fine to plug in a lock-step BA protocol since the poor latency of lock-step BA does not affect the \latency.
Protocol \ba proceeds in a non-lock-step fashion under the good-case.
With foresight, we also note that when we apply this protocol to implement SMR, the BA will be replaced by (roughly speaking) subsequent iterations of Step \ref{ba:step:forward}--\ref{ba:step:commit} where each iteration uses a new leader.

\paragraph{When a clock skew $\clock$ exists.}
Now we show that how to extend the results to the case when there exists a clock skew $\clock$ so that replicas may start the protocol at times at most $\clock$ apart from each other.
Since any replica may start the protocol at most $\clock$ earlier than any other replica, it will receive $f+1$ \vote messages within $3\Delta+\clock$ in the good-case, and therefore the forwarded \vote messages reaches all honest replicas within $4\Delta+\clock$ at their local time. 
Therefore, the parameter $3\Delta$ in the conditions at Step \ref{ba:step:commit} is replaced with $3\Delta+\clock$, and the parameter $4\Delta$ at Step \ref{ba:step:ba} is replaced with $4\Delta+\clock$.
Due to clock skew, replicas may invoke the BA at Step \ref{ba:step:ba} at times at most $\clock$ away from each other. Therefore, the BA primitive also need to tolerate up to $\clock$ clock skew.
For instance, any lock-step BA can do so by setting each round duration to be $2\Delta$ and use a clock synchronization algorithm \cite{dolev1995dynamic, abraham2019synchronous} to enforce the lock-step synchrony.
Since the value of $\clock$ is unknown, we can assume $\clock=\Delta$ as the worst-case clock skew.

\subsection{Correctness of Protocol \ba}
\label{sec:ba:correctness}

\begin{lemma}\label{lem:ba:safety}
    If an honest replica commits value $b$ at Step~\ref{ba:step:commit}, then 
    (i) no honest replica commit any other value $b'\neq b$ at Step~\ref{ba:step:commit}, and
    (ii) all honest replicas set $\locked=b$ before invoking the Byzantine agreement at Step~\ref{ba:step:ba}.
\end{lemma}

\begin{proof}
    Suppose an honest replica commits $b$ at Step~\ref{ba:step:commit} after receiving $f+1$ \vote messages.
    Then at least one honest replica $h$ forwards the proposal $B$ containing $b$ at time $t$, and sends the \vote message at time $t+\Delta$.
    
    First, we show that for any value $b'\neq b$, there does not exist $f+1$ distinct signed \vote messages for proposal $B'$ that contains $b'$.
    For the sake of contradiction, suppose there exists $f+1$ \vote messages for $B'$.
    Then at least one \vote for $B'$ comes from an honest replica $h'$. 
    Let $t'$ denote the time when $h'$ forwards the proposal.
    If $t'\leq t$, then the forwarded proposal $B'$ from $h'$ containing the value $b'$ will reach $h$ no later than $t'+\Delta\leq t+\Delta$, which will prevent $h$ from sending the \vote message, a contradiction.
    If $t'> t$, then similarly the forwarded proposal $B$ from $h$ containing the value $b$ will reach $h'$ no later than $t+\Delta<t'+\Delta$, which will prevent $h'$ from sending the \vote message, a contradiction.
    The same argument applies for any honest replica that may have voted for proposal $B'$.
    Hence, there does not exist $f+1$ distinct signed \vote messages for proposal $B'$ that contains value $b'\neq b$, and no honest replica commits to any other value $b'\neq b$ at Step~\ref{ba:step:commit}.

    For part (ii), since an honest replica $h$ commits $b$ at time $t\leq 3\Delta$, and forwards the $f+1$ votes for proposal $B$ that contains $b$ to all replicas, all honest replicas receives $f+1$ votes no later than time $4\Delta$.
    Since there is no $f+1$ votes for other proposal that contains value $b'\neq b$, all honest replicas set $\locked=b$ before invoking the Byzantine agreement in Step~\ref{ba:step:ba}.
\end{proof}

\begin{theorem}\label{thm:ba}
    \ba protocol solves  Byzantine agreement in the synchronous authenticated setting with near-optimal \latency of  $\Delta+2\delta$ and optimal resilience of $f<n/2$.
\end{theorem}

\begin{proof}
    {\bf Agreement.}
    If all honest replicas commit at Step~\ref{ba:step:ba}, then due to the agreement property of the Byzantine agreement primitive, all honest replicas commit on the same value.
    Otherwise, there must be some honest replica that commits at Step~\ref{ba:step:commit}.
    Let $h$ denote the first honest replica that commits, and let $b$ denote the committed value.
    By Lemma \ref{lem:ba:safety}, no other value is committed by any honest replica at Step~\ref{ba:step:commit}, and all honest replicas set $\locked=b$ before invoking the Byzantine agreement primitive at Step~\ref{ba:step:ba}.
    Therefore, when invoking the BA primitive at Step~\ref{ba:step:ba}, the inputs of all honest replicas are the same. Then by the validity condition of the BA primitive, the output of the agreement is also $b$. Any honest replica that does not commit at Step~\ref{ba:step:commit} will commit on value $b$ at Step~\ref{ba:step:ba}. 
    
    {\bf Termination.}
    According to the protocol, honest replicas terminate at Step~\ref{ba:step:ba}, and they commit a value before termination.
    At time $4\Delta$, all honest replicas invoke an instance of Byzantine agreement. Termination follows from the termination property of this Byzantine agreement instance.
    
    {\bf Validity.}
    If all honest replicas initially have the same value $b$ as input, then we show that all honest replicas are able to commit $b$ at Step~\ref{ba:step:commit}.
    After Step~\ref{ba:step:exchange}, all honest replicas receive $f+1$ signed messages of value $b$ and form the same proposal $B$ after at most $\Delta$ time.
    Since there are at most $f$ Byzantine replicas, no other proposal $B'$ containing $f+1$ signed message for another value $b'\neq b$ can be formed.
    Therefore, during the $\Delta$ waiting period before Step~\ref{ba:step:vote}, no other proposal $B'$ will be received by any honest replica, and all honest replicas send the \vote messages at Step~\ref{ba:step:vote}. 
    Finally, the \vote messages reach all honest replicas after at most $\Delta$ time, leading to all honest replicas to commit on the value $b$ within $3\Delta$ time at Step~\ref{ba:step:commit}.
    
    {\bf Good-case latency.}
    In the good case, all honest replicas have the same input $b$.
    Therefore, all honest replicas form the same proposal at $\delta$ after receiving each other's value, and no other conflicting proposal can be formed. Then, all honest replica will wait for time $\Delta$ before sending the \vote message. Next, the above \vote messages reach all honest replicas after $\delta$ time,
    and all honest replicas commit on the sender's proposal at time $\Delta+2\delta$.
\end{proof}

\subsection{\bb under Synchrony}
\label{sec:bb}

The \ba protocol can be easily extended to Byzantine broadcast with small modifications. 
Recall that the \latency for Byzantine broadcast is defined when the designated sender is honest.
The protocol is given in Figure \ref{fig:bb:sync}. 
The main differences are (i) at Step~\ref{bb:step:propose}, the designated sender broadcasts its input, instead of the input exchange among all replicas in \ba, and (ii) the proposal at Step 2-4 contains only the value signed by the sender from Step~\ref{bb:step:propose}, instead of $f+1$ values signed by replicas.
As a result, the proofs are analogous to \ba and we omit them.
Note that \bb has resilience $f<n/2$, which is not optimal.

\begin{figure}[h]
    \centering
    \begin{mybox}
    Initially, every replica $i$ starts the protocol at the same time $t=0$, and sets $\locked=\bot$.
\begin{enumerate}[itemsep=0pt,topsep=0pt]
    \item\label{bb:step:propose} \textbf{Propose.} The designated sender $L$  with input $b$ sends $\langle \texttt{propose}, b\rangle_L$ to all  replicas.
    
    \item Step 2-5 is identical to Step~\ref{ba:step:forward}-\ref{ba:step:ba} of \ba, except that the proposal contains only one value signed by the leader instead of $f+1$ signed values, and replica detects conflict when receiving two different proposals signed by the leader at Step 3.
    
    
    

    
    
\end{enumerate}
    \end{mybox}
    \caption{\bb Protocol under the Synchronous Model}
    \vspace{-1em}
    \label{fig:bb:sync}
\end{figure}

\begin{theorem}\label{thm:ba}
    \bb protocol solves  Byzantine broadcast in the synchronous authenticated setting with near-optimal \latency of  $\Delta+2\delta$.
\end{theorem}

\section{State Machine Replication with Good-case Latency of $\Delta+2\delta$}
\label{sec:smr}

In this section, inspired by the techniques from protocol \ba and \bb, we construct a Byzantine fault tolerant state machine replication protocol \smr that has a \latency of $\Delta+2\delta$ and optimal resilience of $f<n/2$.

\smr protocol (Figure \ref{fig:smr:sync}) takes the stable leader approach that proceeds in views, which consists a steady state, and a view-change to replace a Byzantine leader.
In the steady state, the leader of the current view is responsible for making progress, where view is a integer that increments after each view-change. 
The leader of a view can simply be elected by round-robin, i.e., the replica $(v \mod n)$ where $v$ is the view number and $n$ is the number of replicas.
If the leader behaves maliciously or does not make progress, the replicas will blame the leader and start the view-change protocol to replace the leader.

We introduce some terminology that will be used in our protocol.

{\bf Block format, block extension and equivocation.}
Clients' requests are batched into blocks, and the protocol outputs a chain of blocks $B_1,B_2,...,B_k,...$ where $B_k$ is the block at height $k$.
Each block $B_k$ has the following format $B_k=(b_k, h_{k-1})$ where $b_k$ is a batch of new client requests and $h_{k-1}=H(B_{k-1})$ is the hash digest of the previous block at height $k-1$.
We say that a block $B_l$ {\em extends} another block $B_k$, if $B_k$ is an ancestor of $B_l$ according to the hash chaining where $l\geq k$.
We define two blocks $B_l$ and $B'_{l'}$ to be {\em equivocating}, if they are not equal and do not extend on another. 
As it will become clear in the protocol, we also define two proposals to be {\em equivocating}, if they are proposed and signed by the leader in the same view and both have $\fs\neq \bot$.
The block chaining simplifies the protocol in the sense that once a block is committed, its ancestors can also be committed.
    
{\bf Certificate, certificate ranking.}
A quorum certificate is a set of signatures on a block by a quorum of replicas, which consists of $f+1$ replicas out of $2f+1$ replicas for \smr.
We use $\fc_v(B_k)$ to denote a certificate for $B_k$ in view $v$, consisting of $f+1$ distinct signed \vote messages for block $B_k$. 
Certified blocks are ranked first by the views that they are created and then by their heights, that is, blocks with higher views have higher ranks, and blocks with higher height have higher ranks if the view numbers are equal.
We use $\fb_v$ to denote a blame certificate in view $v$, consisting of $f+1$ distinct \blame messages in view $v$.

Since the clients' requests are batched into blocks, the BFT SMR protocol achieves {\em safety} if honest replicas always commit the same block $B_k$ for each height $k$, and {\em liveness} if all honest replicas keep committing new blocks.
When it is clear in the protocol, a replica broadcasting a message means it sending the message to all other replicas; it does not mean invoking a Byzantine broadcast instance.

\subsection{\smr Protocol under Synchrony}

\begin{figure}[h!]
    \centering
    \begin{mybox}
    \textbf{Steady State Protocol for Replica $i$}

Let $v$ be the current view number and replica $L$ be the leader of the current view. The leader proposes a block every $\alpha$ time, where $\alpha$ is a parameter.
\begin{enumerate}
    \item\label{smr:propose} {\bf Propose.} The leader $L$ sends $\langle \texttt{propose}, B_k, \fs, v \rangle_L$ to all other replicas, 
    where $B_k=(b_k,h_{k-1})$ is a height-$k$ block, containing a batch of new client requests $b_k$ and a hash digest $h_{k-1}=H(B_{k-1})$ of a height-$(k-1)$ block $B_{k-1}$.
    For the first proposal in a new view after a view-change,
    $\fs$ is a set of $f+1$ distinct signed \status messages received during the view-change, 
    and $B_{k-1}$ is the highest certified block among blocks in $\fs$.
    Otherwise, $\fs=\bot$ and $B_{k-1}$ is the last block proposed by $L$.
    
    \item\label{smr:forward} {\bf Forward.} Upon receiving a valid new proposal $\langle \texttt{propose}, B_k, \fs, v \rangle_L$,
    forward the proposal to all other replicas.
    Once one of the following conditions is true: 
    (1) $\fs$ contains $f+1$ distinct signed \status messages, $B_k$ extends the highest certified block in $\fs$, and $B_k$ extends a chain known to the replica, or
    (2) $\fs=\bot$, $B_k$ extends the highest certified block known to the replica, and $B_k$ extends a chain known to the replica, perform the following.
    For any block in the chain that $B_k$ extends, broadcast the block if the replica hasn't forwarded it yet.
    Set $\texttt{vote-timer}_k$ to $\Delta$ and start counting down.

    \item\label{smr:vote} {\bf Vote.} When $\texttt{vote-timer}_k$ reaches $0$, if no equivocating blocks/proposals signed by $L$ or a blame certificate $\fb_v$ are received, send a \vote to all other replicas in the form of $\langle \vote, B_k, v \rangle_i$.
    
    \item\label{smr:commit} {\bf Commit.} 
    Upon receiving $f+1$ distinct signed \vote messages for block $B_k$, if no equivocating blocks/proposals signed by $L$ or a blame certificate $\fb_v$ are received, broadcast the \vote certificate $\fc_v(B_k)$ containing $f+1$ votes, and commit $B_k$ with all its ancestors.
\end{enumerate}

\textbf{View-change Protocol for Replica $i$}

Let $L$ and $L'$ be the leader of view $v$ and $v+1$ respectively.
\begin{enumerate}
    \item\label{smr:blame} {\bf Blame.} If less than $p$ valid blocks are committed in $6\Delta+(p-1)\alpha$ time in view $v$, broadcast $\langle \blame, v\rangle_i$. If equivocating blocks/proposals signed by $L$ are received, broadcast $\langle \blame, v\rangle_i$ and the two equivocating blocks/proposals.
    
    \item\label{smr:newview} {\bf New-view.}
    Upon receiving $f+1$ distinct signed $\langle \blame, v\rangle$ messages, form a \blame certificate $\fb_v$, broadcast $\fb_v$, abort all \texttt{vote-timer}(s) and stop sending \vote in view $v$.
    Wait for $2\Delta$ time and enter view $v+1$.
    Upon entering view $v+1$, send the new leader $L'$ a $\langle \status, B_{k'}, \fc_{v'}(B_{k'}), v \rangle_i$ message where $\fc_{v'}(B_{k'})$ is certificate for the highest certified block $B_{k'}$.
    (Ignore any block certified in view $\leq v$ that is received after this point.)
    If the replica is the new leader $L'$, wait for another $2\Delta$ time upon entering view $v+1$.
    
\end{enumerate}
    \end{mybox}
    \caption{\smr under the Synchronous Model}
    \label{fig:smr:sync}
\end{figure}

Here we describe the protocol \smr in Figure \ref{fig:smr:sync}.
In the steady state of the protocol, the leader $L$ of the current view $v$ can propose blocks chained by block hashes every $\alpha$ time, where $\alpha >0$ is a predefined parameter.
The steady state of \smr is similar to Step 1-4 of protocol \bb, and uses similar techniques such as proposal forwarding and the $\Delta$ waiting time before voting.
For the first proposal after a view-change, the leader proposes a block that extends the highest certified block known to the leader (from the status messages received). For other proposals, each proposed block extends the last proposed block (Step \ref{smr:propose}).
We assume that the honest replicas are always able to generate and propose valid proposal blocks.
Once any replica receives a new valid proposal from the leader $L$, it forwards the proposal (Step \ref{smr:forward}) for equivocation check. 
Due to the non-lock-step nature of our protocol, it is possible that a replica receives the proposal from the re-proposal of some other replica first. In this case, the replica also accept the proposal if it is valid. 
Then, once the replica has all blocks in the chain that the proposal block is extending, and the proposal block extends the highest certified block, it locally starts a timer called $\texttt{vote-timer}$ to wait for $\Delta$ time (Step \ref{smr:forward}).
When the timer expires and the replica does not receive any equivocating block signed by the leader $L$ or a blame certificate (containing $f+1$ \blame messages), it broadcast a \vote message in the form of $\langle \vote, B_k, v \rangle$ (Step \ref{smr:vote}). 
Once the replica gathers $f+1$ distinct \vote messages to form a certificate for the block $B_k$ without receiving a blame certificate (containing $f+1$ \blame messages) or equivocating blocks, it forwards the \vote certificate $\fc_v(b)$, commits the block $B_k$ and all its ancestor blocks (Step \ref{smr:commit}). The forwarding of \vote certificate notifies other replicas to update their highest certified block before entering the next view. 

When the leader is Byzantine, it can deviate from the protocol by either stalling or equivocating \footnote{In practice, the Byzantine leader may also censor clients' requests, i.e., always proposes blocks with its own requests or transactions. Any BFT SMR protocol should have a scheme to prevent such censoring. For our protocol, we can either do leader rotations regularly, or add another blame rule for censoring to replace the Byzantine leader. }. 
To ensure liveness, the view-change protocol will be triggered when a quorum ($f+1$ out of $2f+1$) of replicas discover the malicious behaviors of the current leader. 
If the current leader does not keep proposing valid blocks quick enough, any replica that does not commit blocks in time will blame the leader by sending a \blame message. The time period of $6\Delta+(p-1)\alpha$ is sufficient for $p$ proposals from an honest leader to be committed (see the proof of Theorem \ref{thm:smr:liveness}). 
Note that the above time period only occur during the view-change, and does not affect the \latency of the protocol when a honest leader is in charge.
If the leader equivocates, when any replica receives the equivocating blocks, it will also blame the leader by broadcasting a \blame message and the pair of equivocating blocks to help other replicas detect the equivocation. 
Once the replica collects $f+1$ \blame messages, it forms and broadcasts a \blame certificate, aborts the timer and stops voting in view $v$. 
Then the replica waits for a $2\Delta$ period before entering the next view $v+1$. The $2\Delta$ waiting period is for receiving the \vote certificate from other replicas. 
Then, when entering the new view, the replica sends the certificate of the highest known certified block to the new leader, and returns to the steady state. If a replica is the new leader, it needs to wait for another $2\Delta$ time period in order to receive the status messages sent by all honest replicas during view-change.

\subsection{Correctness of Protocol \smr}

We say a block $B_k$ is committed directly, if the replica commits $B_k$ at Step \ref{smr:commit} by receiving $\fc_v(B_k)$ but no equivocating blocks or blame certificate $\fb_v$.
A block $B_k$ is committed indirectly, if $B_k$ is committed because it is the ancestor of a directly committed block.


\begin{lemma}\label{lem:smr:safety}
    If an honest replica directly commits a block $B_l$ in view $v$, 
    then a certified block that ranks no lower than $\fc_v(B_l)$ must equal or extend $B_l$.
\end{lemma}

The proof of Lemma \ref{lem:smr:safety} is is very similar to that of 
Lemma~\ref{lem:ba:safety} and is presented in Appendix \ref{app:proof:smr:safety} due to lack of space.

\begin{theorem}[Safety]\label{thm:smr:safety}
    Honest replicas always commit the same block $B_k$ for each height $k$.
\end{theorem}

\begin{proof}
    Suppose two blocks $B_k$ and $B_k'$ are committed at height $k$ at any two honest replicas. Suppose $B_k$ is committed due to $B_l$ being directly committed in view $v$, and $B_k'$ is committed due to $B_{l'}'$ being directly committed in view $v'$.
    Without loss of generality, suppose $v\leq v'$, and for $v=v'$, further assume that $l\leq l'$. 
    Since $B_l$ is directly committed and $B_{l'}$ is certified and ranks no lower than $\fc_v(B_l)$,
    by Lemma \ref{lem:smr:safety}, $B_{l'}$ must equal or extend $B_l$. Thus, $B_k'=B_k$.
\end{proof}

\begin{restatable}[Liveness]{theorem}{smrliveness}
\label{thm:smr:liveness}
All honest replicas keep committing new blocks.
\end{restatable}


The proof of Theorem \ref{thm:smr:liveness} is presented in Appendix \ref{app:proof:smr:liveness} due to lack of space.

\begin{theorem}[Good-case Latency]\label{thm:smr:efficiency}
    In the good-case, every proposed block will be committed in $\Delta+2\delta$ time after being proposed.
\end{theorem}

\begin{proof}
    In the good-case, the leader is honest, and its proposal will take time $\delta$ to reach all replicas. Then, all honest replica will wait for time $\Delta$ before sending the \vote message. Finally, the above \vote messages reach all honest replicas after $\delta$ time, leading to commit at all honest replicas.
    Therefore, the any proposed block will be committed in $\Delta+2\delta$ time if the leader is honest.
\end{proof}



\section{The Mobile Link Failure Model}
\label{sec:mlf:summary}

Our previous protocols \ba, \bb and \smr heavily relies on the synchrony assumption that messages between honest replicas are delivered within $\Delta$ time.
Once this assumption is violated even at only one honest replica, our protocols as well as most of the synchronous BFT protocols would fail under certain executions.
To make our protocols more practical, we further strengthen them by considering the {\em  mobile link failures} suggested in the literature \cite{schmid2009impossibility, biely2011synchronous}. 

The mobile link failure captures the case where a certain fraction of the communication links at any honest replica are controlled by the adversary so that the messages via these link may be lost or delayed. 
We first generalize the mobile link failure model for non-lock-step synchrony, where synchronous protocols do not execute in a lock-step fashion. Then, we propose a generic transformation to make a synchronous protocol tolerate mobile link failures. Applying the above transformation to our previous protocols, we can obtain protocols for BA, BB and BFT SMR under the mobile link failure model with \latency of $2\Delta+4\delta$.
Lastly, we prove that $2\Delta$ is a lower bound for the \latency of BA/BB under the mobile link failure model (Section \ref{sec:ml:lb}).

\paragraph{The Mobile Link Failure Model without Lock-step Execution.}
\label{sec:ml:model}
We model the communication channel between any two replicas as two directed links,
and use $(i,j)$ to denote the directed link from replica $i$ to replica $j$. 
Link $(i,j)$ is called a send link of replica $i$, and a receive link for replica $j$. 
Each replica has $n$ send links and $n$ receive links, including a link that connects itself.

Prior works \cite{schmid2002formally, bielyoptimal, weiss2001consensus, schmid2009impossibility, biely2011synchronous} consider mobile link failures under lock-step synchrony: for each lock step (or round), a link $(i,j)$ is faulty if replica $j$ does not receive the message sent by replica $i$ in this round.
The link failure is mobile in the sense that the set of faulty links can change at each round.
Since our protocol does not require lock-step execution, we generalize the mobile link failure model for non-lock-step synchrony as follows. 

Under our non-lock-step model, each link is either {\em faulty} or {\em non-faulty} at any time point, and two directions of a link between any two replicas may fail independently, i.e., link $(i,j)$ may be non-faulty while link $(j, i)$ is faulty. 
If a link $(i,j)$ is non-faulty at time $t$, then the message sent by replica $i$ at time $t$ will be received by replica $j$ at time $\leq t+\delta$.
Otherwise, the message may be lost or delayed.
Link failures are {\em mobile} in the sense that an adversary can control the set of faulty links, subject to the following constraints.
\begin{enumerate}[noitemsep,topsep=0pt]
    \item\label{mobility:1} The adversary can corrupt up to $f_l^s$ send links and $f_l^r$ receive links at any replica, as long as $f_l^s+f_l^r < n-f$.
    \item\label{mobility:2} If a link turns from faulty to non-faulty at time $t$, it continues to count towards the threshold of faulty links until $t+\delta$.
\end{enumerate}

We explain the necessity of the above two constraints.
Constraint (\ref{mobility:1}) is necessary for solving BB or BA under even static link failures \cite{schmid2009impossibility}, and hence also necessary under our definition of mobile link failures.
Constraint (\ref{mobility:2}) states that a recovered link remains on the adversary's corruption budget for an additional $\delta$ time.
Some constraint in this vein is necessary;
otherwise, the adversary can drop all messages with a budget of a single link failure.
Since we assume no lock-step execution, it is normal that no two honest nodes ever send messages at precisely the same time. 
Then, whenever there is a message sent via a link at time $t$, the adversary can corrupt the link at time $t$ and immediately uncorrupt it.
Also note that constraint (\ref{mobility:2}) is consistent with the mobile link failure model under lock-step synchrony. In the lock-step model, the set of faulty links remain unchanged during a round, which has duration at least $\delta$.
Therefore, the mobile link failure model for lock-step synchrony \cite{schmid2009impossibility, biely2011synchronous} is a special case of our generalized definition above.


\paragraph{A generic transformation to tolerate mobile link failures under non-lock-step synchrony.}
\label{sec:ml:transformation}
In this section, we show that given a Byzantine fault tolerant protocol under the synchronous model, there is a generic approach to transform it into a protocol that handles mobile link failures as defined in the previous section.
In fact, it has been observed that under lock-step synchrony,
a two-round simulation under the mobile link failure model can implement a one-round multicast over non-faulty links under the synchronous model (Corollary 1 of \cite{schmid2009impossibility}).

Here, we present a similar transformation for non-lock-step synchronous protocols: (i) when any replica receives a message described in the protocol, it forwards the message only once to all other replicas except the sender of the message,
(ii) every timing parameter in the protocol is doubled, i.e., if the original protocol waits for $T$ time at some step, then the transformed protocol waits for $2T$ at that step.


The following lemma shows that after the generic transformation above, any message sent by an honest replica will be received by all honest replicas after $2\delta$, implying a $2\delta$-delay simulation of the synchronous model.
Therefore, a protocol for the mobile link failure model can be obtained from a protocol designed for the synchronous model.

\begin{lemma}\label{lem:ml:echo}
    If an honest replica sends a message at time $t$, then all honest replicas receive the message by time $t+2\delta$.
\end{lemma}

\begin{proof}
    Suppose an honest replica $h_s$ sends a message at time $t$ and another honest replica $h_t$ does not receive the message by time $t+2\delta$.
    Let $A$ be the set of honest replicas that are connected to non-faulty send links of $h_s$ at time $t$.
    Since the number of faulty send links at $h_s$ is at most $f_l^s$, $|A|\geq n-f-f_l^s$.
    According to the model of mobile link failures in Section \ref{sec:ml:model}, replicas in $A$ receive the message by time $\leq t+\delta$. 
    Each replica in $A$ tries to relay the message to $h_t$ some time between $t$ and $t+\delta$.
    In order for $h_t$ not to receive the message by $t+2\delta$, all these $|A|$ attempts have to fail.
    But in this case, by constraint (\ref{mobility:2}), all these $|A|\geq n-f-f_l^s$ receive links of $h_t$ will be counted as faulty at time $t+\delta$. 
    This violates constraint (\ref{mobility:1}), which states that less than $n-f-f_l^s$ receive links can be counted as faulty at any replica at any time.
\end{proof}

\paragraph{Lower bound on good-case latency under mobile link failures.}
\label{sec:ml:lb}
Next, we show a lower bound of $2\Delta$ on the \latency for any Byzantine agreement protocol that tolerates mobile link failures. 
In fact, we can prove the lower bound with static link failures.
The lower bound holds for Byzantine broadcast as well, by an analogous proof.

\begin{theorem}
\label{thm:ml:lb}
    Any Byzantine agreement protocol in the synchronous authenticated setting that (i) is resilient to $\geq \frac{n+1}{3}$ Byzantine faults and (ii) tolerates at most $f_l^s$ static send link failures and $f_l^r$ static receive link failures at each replica where $f_l^s+f_l^r<n-f$, must have a \latency of at least $2\Delta$.
\end{theorem}

\begin{proof}
    Suppose for the sake of contradiction that such a protocol exists.
    Let there be $n=3f-1$ replicas with $f$ Byzantine replicas.
    Divide the replicas into three groups $A,B,C$, where $|A|=|C|=f-1$ and $|B|=f+1$. 
    All replicas are connected to each other, except that replicas in $A$ are disconnected to replicas in $C$ due to the link failure.
    The construction satisfies the link failure requirement since $|A|+|C|=2f-2< n-f=2f-1$.
    Consider the following three scenarios.
  
    In scenario 1, it satisfies the good-case, and all honest replicas have input $0$. All $f-1$ replicas in $C$ are Byzantine. 
        There is one Byzantine replica $b$ in $B$, and the other $f$ replicas in $B$ are honest.
        The Byzantine replica in $C$ remain silent.
        The Byzantine replica $b$ follows the protocol except the following.
        (1) The replica $b$ pretends that no message is received from the leader.
        (2) The replica $b$ pretends that no message is received from the remaining honest replicas in $B$, and $b$ sends no message to the honest replicas in $B$.
        (3) For any messages received from the replicas in $A$, replica $b$ pretends that the messages are received after a $\Delta$ network delay (suppose Byzantine replicas know the actual network delay $\delta$, which is $0$ in the good-case). Replica $b$ also intentionally delays any messages it sends, to pretend that the network delay is $\Delta$.
        According to the validity condition, the honest replicas in $A$ and $B$ will commit $0$ within $<2\Delta$ time.
        
        Scenario 2 is the mirror case of scenario 1 and satisfies the good-case, specified as follows.
        All honest replicas have input $1$. All $f-1$ replicas in $A$ are Byzantine and one replica $b$ in $B$ is Byzantine.
        The Byzantine replicas in $A$ remain silent.
        The Byzantine replica $b$ follows the protocol except the following.
        (1) The replica $b$ pretends that no message is received from the leader.
        (2) The replica $b$ pretends that no message is received from the remaining honest replicas in $B$, and $b$ sends no message to the honest replicas in $B$.
        (3) For any messages received from the replicas in $C$, replica $b$ pretends that the messages are received after a $\Delta$ network delay (suppose Byzantine replicas know the actual network delay $\delta$, which is $0$ in the good-case). Replica $b$ also intentionally delays any messages it sends, to pretend that the network delay is $\Delta$.
        According to the validity condition, the honest replicas in $C$ and $B$ will commit $1$ within $<2\Delta$ time.
        
        In scenario 3, $f$ replicas in $B$ are Byzantine, and the remaining one replica $h$ in $B$ is honest. 
        Replicas in $A$ have input value $0$, replicas in $C$ have input value $1$, and the replica $h$ has no input value.
        The Byzantine replicas in $B$ behave to replicas in $A$ exactly the same as the honest replicas in $B$ in scenario 1, behave to replicas in $C$ exactly the same as the honest replicas in $B$ in scenario $2$, and send no message to the honest replica $h$ in $B$.
        Suppose that the network delay between honest replicas is $\delta=\Delta$ in scenario 3, then we can claim that replicas in $A$ cannot distinguish scenario 1 and 3 within time $<2\Delta$.
        \begin{itemize}
            \item First, to replicas in $A$, the Byzantine replica $b$ in $B$ behaves in scenario 1 exactly the same as the honest replica $h$ in $B$ behaves in scenario 3.
            In scenario 3, the network delay between honest replicas in $B$ and $A$ is $\Delta$, only messages sent by $B$ at time $< \Delta$ will be received by replicas in $A$ before they commit. Since any message from $C$ reaches the honest replica $h$ in $B$ at time $\geq \Delta$ in scenario 3, those messages from $C$ do not influence the messages that $h$ sends to $A$ that are received within time $<2\Delta$.
            In scenario 1, the replicas in $C$ are silent, which also do not influence the messages that $b$ sends to $A$ before the replicas in $A$ commit.
            Also, the Byzantine replica $b$ in $B$ in scenario 1 pretends that no message is received from the sender or from the honest replicas in $B$, and it pretends that the communication with replicas in $A$ has the maximum network delay $\Delta$. 
            Therefore, to replicas in $A$, the Byzantine replica $b$ in $B$ in scenario 1 behaves exactly the same as the honest replica $h$ in $B$ in scenario 3.
            
            \item Also by construction, to replicas in $A$, the $f$ Byzantine replicas in $B$ in scenario 3 behave exactly the same as the $f$ honest replicas in $B$ in scenario 1. 
        \end{itemize}
        Therefore, replicas in $A$ cannot distinguish scenario 1 and 3, and they will commit $0$ in scenario 3 within time $<2\Delta$.
        Similarly, replicas in $C$ cannot distinguish scenario 2 and 3. Therefore the replicas in $C$ will commit $1$ in scenario 3 within time $<2\Delta$.
        However, this violates agreement since honest replicas commit different values. 
\end{proof}

\paragraph{The mobile sluggish fault model.}
\label{sec:msf:summary}
In the mobile sluggish fault model \cite{chan2018pili, guo2019synchronous, synchotstuff}, an honest replica is either {\em sluggish} or {\em prompt} at a given time.
The messages sent or received by sluggish replicas may be delayed in the network, while the prompt replicas follow the message delay bound. 
The set of sluggish replicas is mobile in the sense that it can change at any instance of time.
We present the results for the mobile sluggish fault model in Appendix \ref{app:msf}.

Formally, if an honest replica $r_1$ is prompt at time $t_1$, then any message sent by $r_1$ at time $\leq t_1$ will be received by any honest replica that is prompt at time $t_2\geq t_1+\Delta$. 
Intuitively, in the mobile sluggish model, any message sent by a sluggish replica would satisfy the maximum message delay bound when the replica becomes prompt. 
We assume the number of honest and prompt replicas is always $>n/2$ at any time, which is necessary to solve Byzantine agreement or broadcast under the mobile sluggish fault model \cite{guo2019synchronous, synchotstuff}.

Under the mobile sluggish fault model introduced above, we can increase the robustness of our protocols \ba, \bb and \smr to tolerate mobile sluggish faults using a technique that adds extra communication rounds in the protocol, while keeping the \latency near-optimal.

\paragraph{Comparison of the two models.}
Despite similarities of the mobile link failures and mobile sluggish faults above, there are differences between these two models, including the following.
\begin{enumerate}[noitemsep,topsep=0pt]
    \item With mobile sluggish faults, an honest but sluggish replica cannot communicate with any other replicas before it becomes prompt, whereas in mobile link failure model any honest replicas always have a certain number of non-faulty links to send/receive messages.
    \item The mobile sluggish model assumes the set of replicas that are Byzantine or sluggish must be minority at any time, whereas in the mobile link failure model, the number of faulty send and receive links must be less than $n-f$ at any honest replica.
    \item The set of sluggish replicas can change at any instance of time in mobile sluggish model, while the set of faulty link in mobile link failure model can also change dynamically subject to some constraints (see Section \ref{sec:mlf:summary}).
\end{enumerate}

\section{Related Work}

\stitle{Byzantine fault tolerant protocols.}
Byzantine fault tolerant protocols, first proposed by Lamport \cite{lamport1982byzantine}, have received significant amount of attention for several decades.
For Byzantine broadcast, Dolev-Strong protocol \cite{dolev1983authenticated} is a deterministic $f+1$-round protocol with $O(n^2f)$ communication complexity.
The Dolev-Strong protocol can also be modified to solve authenticated BA for the $t<n/2$ case with the same round and communication complexity.
A sequence of effort has been made on randomized protocol for reducing the round complexity and message complexity \cite{ben1983another, rabin1983randomized, katz2006expected}, and most efficient solutions for both Byzantine agreement and broadcast are proposed by Abraham et al. \cite{abraham2019synchronous} with expected constant round and expected quadratic communication complexity.
The state-of-the-art BFT SMR protocol Sync HotStuff \cite{synchotstuff} improves the \latency to $2\Delta$ under the synchronous model.
As a comparison, our protocol improves the \latency to $\Delta+2\delta$, which is near-optimal for both cases.
An earlier work~\cite{attiya1991bounds} proposes a synchronous crash-tolerant agreement protocol with latency $(2f-1)\delta+\Delta$, and hence \latency $\Delta$.
Under partial synchrony, a latency-optimal transformation from Byzantine consensus to BFT SMR is proposed in \cite{sousa2012byzantine} for building practical BFT SMR systems \cite{bessani2014state}.
Another line of research aims at developing BFT protocols with small latencies in the optimistic cases when certain conditions are satisfied, under synchrony \cite{pass2018thunderella, synchotstuff} or asynchrony \cite{dutta2005best, martin2006fast, song2008bosco}.
In particular, the notion of optimistic responsiveness \cite{pass2018thunderella, synchotstuff} is proposed for synchronous BFT SMR protocols, where the protocol can make decisions in network speed $O(\delta)$ when $>3n/4$ votes are collected.
It should be noted that our paper investigates the upper and lower bounds for \latency defined under any adversary strategies, instead of the optimistic case.

\stitle{Weaker network models.}
{\em The mobile link failure model.}
For models that has no restriction on the number of failed links, Santoro and Widmayer \cite{santoro1989time} show that the consensus is unsolvable with even a single process suffering from such link failure.
Thus, the authors in \cite{schmid2002formally} introduce a mobile link failure model with constraints on the number of send and receive link failures.
There has been a sequence of efforts on adapting existing consensus algorithms for the mobile link failure model \cite{bielyoptimal, weiss2001consensus}, and results on the lower bounds on the required number of processes and rounds \cite{schmid2009impossibility}.
Another line of work identifies the tight necessary and sufficient condition on the underlying communication graph for solving iterative approximate consensus under mobile link failures \cite{tseng2014iterative}.
The model and results above assume protocols of lock-step execution, and our results do not pose such an assumption.
{\em The mobile sluggish fault model.}
Recently, Guo et al. consider a new model that allows the bound $\Delta$ on the message delay to be violated for a set of honest replicas under synchrony, to better capture the reality when some honest replicas are partitioned or offline due to network misbehavior.
This type of faults is later called sluggish, and considered in both PiLi \cite{chan2018pili} and Sync HotStuff \cite{synchotstuff} to introduce more fault-tolerance to the BFT SMR protocols.
In this paper, we introduce techniques to handle mobile link failures and mobile sluggish faults for our BA, BB and BFT SMR protocols.


\section{Conclusion and Future Work}

We propose using the non-lock-step models and the \latency metric for Byzantine agreement and broadcast as they better capture what matters in practical BFT SMR.
We propose the first Byzantine agreement, broadcast, and state machine replication protocols with near-optimal \latency of $\Delta+2\delta$.
We further extend the protocols to weaker models with mobile link failures and mobile sluggish faults, while achieving near-optimal \latency.

The most interesting future work is to close the gap between the upper bound of $\Delta+2\delta$ and the lower bound of $\Delta$ under $n/3\leq f<n/2$ faults. Another intriguing open question is to study the tight bounds of the \latency when the number of faults is not within $[n/3,n/2)$.

The transformation for mobile link failures brings a blowup to communication complexity because each message between a two replicas is now relayed by a group of other replicas.
While it may be difficult to come up with a generic transformation that preserves the communication complexity, it may be possible to directly design protocols in the mobile link failure model to avoid this communication blowup. 

Currently, the mobile link failure model describes slow and lossy links while the mobile sluggish model describe slow nodes (replicas), and there is no clear way to unify them.
An interesting future direction is to come up with an even weaker model that captures both slow/lossy links and slow/lossy nodes, and to design protocols in the unified model.


\newpage

\bibliography{ref}

\appendix

\section{Lower Bound on the Good-case Latency under Synchrony}
\label{app:lb}
We formally restate the lower bound result from \cite{synchotstuff} to a lower bound on the \latency for both Byzantine agreement and broadcast for completeness.

\lowerbound*

\begin{proof}
    We first prove for Byzantine agreement.
    For the sake of contradiction, suppose that there exists such a Byzantine agreement protocol that is resilient to $f\geq n/3$ and have a \latency $<\Delta$.
    Divide the replicas into three groups $P,Q,R$ each of size up to $n/3$. 
    Consider three scenarios as follows.
    In scenario $A$, suppose it satisfies the good-case definition: all honest replicas have input $0$.  Only the replicas in $Q$ are Byzantine and remain silent. 
    By the validity of BA, the protocol ensures that the replicas in $P,R$ commit $0$ in  $<\Delta$ time.
    In scenario $B$, also suppose it satisfies the good-case definition: all honest replicas have input $1$. 
    Only the replicas in $P$ are Byzantine and remain silent.
    By the validity of BA, the protocol ensures that the replicas in $Q,R$ commit $1$ in  $<\Delta$ time.
    In scenario $C$, only the replicas in $R$ are Byzantine. The replicas in $P$ are honest and have input $0$, and the replicas in $Q$ are honest and have input $1$. The Byzantine replicas in $R$ behaves to replicas in $P$ identically as in scenario $A$, and behaves to replicas in $Q$ identically as in scenario $B$.
    Suppose the messages between any replica in $P$ and any replica in $Q$ have latency $\delta=\Delta$.
    Then scenario $A$ and $C$ are indistinguishable to the replicas in $P$ within time $<\Delta$, and they will commit $0$ in $<\Delta$ time as in scenario $A$. 
    Similarly, the replicas in $Q$ will commit $1$ in $<\Delta$ time as in scenario $B$.
    This violates the agreement property of the Byzantine agreement, and therefore such a BA protocol cannot exist.
    
    The proof for Byzantine broadcast is similar. In scenario $A$, the sender is honest and sends $0$, while in scenario $B$ the sender is also honest and sends $1$.
    In scenario $C$, the sender is Byzantine, sends $0$ to replicas in $P$ and $1$ to replicas in $Q$. All messages from the sender are delivered instantaneously, then rest of the proof is identical to that of Byzantine agreement.
\end{proof}

\section{Missing Proofs}
\label{app:missingproofs}

\begin{lemma}\label{lem:smr:chain}
    If a block $B_l$ extending a chain $\fc$ is certified in view $v$, then all honest replicas receive all blocks in $\fc$ before entering view $v+1$.
\end{lemma}
\begin{proof}
    Suppose a block $B_l$ extending a chain $\fc$ is certified in view $v$, then at least one honest replica $h$ sends the \vote message for $B_l$ at time $t_v$ in Step~\ref{smr:vote} in view $v$. 
    Then, at time $t_v-\Delta$ when setting the $\texttt{vote-timer}_l$ in Step~\ref{smr:forward}, $h$ already has the chain $\fc$. Therefore, the honest replica has forwarded all blocks in $\fc$ at time $\leq t_v-\Delta$ according to Step~\ref{smr:forward}.
    Therefore, all blocks in $\fc$ forwarded by $h$ will be received by all honest replicas at time $\leq t_v$. 
    If any honest replica $r$ enters the new view $v+1$ before receiving the blocks in $\fc$, according to Step~\ref{smr:newview} of the protocol, $r$ must have broadcasted the \blame certificate before time $t_v-2\Delta$ due to the $2\Delta$ waiting period before entering the next view. 
    Then, the \blame certificate will reach replica $h$ before time $t_v-\Delta$, which will prevents $h$ from sending the \vote message, causing contradiction.
    Therefore, all honest replicas receive all blocks in $\fc$ before entering view $v+1$.
    
\end{proof}

\subsection{Proof for Lemma \ref{lem:smr:safety}}
\label{app:proof:smr:safety}


\begin{proof}

    Recall that a certified block $B'_{l'}$ with the certificate $\fc_{v'}(B'_{l'})$ ranks no lower than $\fc_{v}(B_l)$ if either (i) $v'=v$ and $l'\geq l$, or (ii) $v'>v$. 
    Suppose that an honest replica directly commits $B_l$ at time $t_c$ in view $v$.
    Then the honest replica receives $f+1$ \vote messages, where at least one of the \vote message is from an honest replica $h$.
    Suppose that $h$ forwards the proposal at time $t$, then it sends the \vote message at time $t_v>t+\Delta$. It is also clear that $t_c>t_v$.
    According to Step~\ref{smr:forward} of the protocol, at time $t_v-\Delta$, $h$ has received all the blocks in the chain that $B_l$ is extending.
    
    First we prove that any block $B'_{l'}$ with $l'\geq l$ certified in view $v$ must equal or extend $B_l$.
    Suppose for the sake of contradiction that some equivocating block $B'_{l'}$ is certified in view $v$. 
    Then at least one \vote for $B'_{l'}$ comes from an honest replica $h'$. 
    Let $t'$ denote the time when $h'$ forwards the proposal, and $t_v'$ denote the time when $h'$ sends the \vote message.
    \begin{itemize}
        \item If $t_v'\leq t_v$, then the forwarded proposal from $h'$ containing the block $B'_{l'}$ will reach $h$ no later than $t'+\Delta<t_v'\leq t_v$. 
        We show that $h$ is able to discover that $B'_{l'}$ equivocates $B_l$ at time $t_v$. 
        Since $h$ receives all the blocks in the chain that $B_l$ is extending at time $t_v-\Delta$, we only need to show that $h$ also receives all the blocks in the chain that $B'_{l'}$ is extending at time $t_v$.
        When the honest replica $h'$ sets its $\texttt{vote-timer}_{l'}$ at time $t_v'-\Delta$, it has all the blocks in the chain $C'$ that $B'_{l'}$ is extending. Then $h'$ has forwarded all blocks in $C'$ at time $t_v'-\Delta$ according to Step~\ref{smr:forward}. 
        Therefore, $h$ should receive all blocks in $C'$ at time $\leq t_v'\leq t_v$, and is able to discover that $B_l$ equivocates with $B'_{l'}$
        This will prevent $h$ from sending the \vote message for $B_l$, causing contradiction.
        
        \item If $t_v'> t_v$, then similarly the forwarded proposal from $h$ containing the block $B_{l}$ will reach $h'$ no later than $t+\Delta<t_v<t_v'$, which will prevent $h'$ from sending the \vote message, also causing contradiction.
    \end{itemize}
    
    Therefore, no other equivocating block $B'_{l'}$ is certified in view $v$.
    
    Now we show that any block $B'_{l'}$ certified in view $v'>v$ must equal or extend $B_l$.
    
    By Lemma \ref{lem:smr:chain}, all honest replicas receive all blocks in the chain $C$ that $B_l$ extends, before entering the new view $v+1$.
    We show that all honest replicas also receive $\fc_v(B_l)$ before entering the new view $v+1$.
    Since the honest replica $h$ directly commits block $B_l$, $h$ also broadcasts the certificate $\fc_v(B_l)$ at Step~\ref{smr:commit} at time $t_c$.
    Therefore,
    all honest replicas receive the certificate $\fc_v(B_l)$ no later than $t_c+\Delta$. Suppose for the sake of contradiction that some honest replica $r$ enters the next view $v+1$ before time $t_c+\Delta$. According to Step~\ref{smr:newview} of the protocol, the honest replica $r$ must have received $f+1$ \blame messages before time $t_c-\Delta$ due to the $2\Delta$ waiting window during view-change. The replica $r$ also broadcasts the \blame certificate before time $t_c-\Delta$ at Step~\ref{smr:newview} according to the protocol. The \blame certificate will reach replica $h$ before time $t_c$ and prevent $h$ from committing, which is a contradiction. Therefore, all honest replicas enter the new view $v+1$ no earlier than $t_c+\Delta$, and receive the certificate $\fc_v(B_l)$ before entering the new view $v+1$. 
    
    Together with the claim that any block $B'_{l'}$ certified in view $v$ must equal or extend $B_l$, the highest certified block at any honest replicas equals or extends $B_l$ when entering view $v+1$.
    If the proposal contains $\fs\neq\bot$, then $\fs$ contains at least one \status message from an honest replica, and thus the block in the proposal must extend $B_l$ in order to get certified.
    If the proposal contains $\fs=\bot$, the honest replicas will only vote for blocks that extends a known chain and the highest certified block, and thus, only blocks that extend $B_l$ can be certified.
    Hence, the highest certified block at any honest replicas still equals or extends $B_l$ in view $v+1$.
    By simple induction, in any future view $v'>v$, the highest certified block at any honest replicas must equal or extend $B_l$, therefore only blocks that equal or extending $B_l$ can be certified.
\end{proof}

\subsection{Proof of Theorem \ref{thm:smr:liveness}}
\label{app:proof:smr:liveness}

\smrliveness*

\begin{proof}

    If the leader is honest, we show that a view-change will not occur and all honest replicas keep committing new blocks.
    By waiting for $2\Delta$ time before entering the new view, an honest leader is able to receive the status messages from all honest replicas, because any honest replica may be at most $\Delta$ later to receive the blame certificate to enter the new view, and the status message takes at most $\Delta$ to reach the leader. 
    Therefore, the honest leader is able to propose a block extending the highest certified block $B$ among $f+1$ distinct signed \status messages as the first proposal.
    By Lemma \ref{lem:smr:chain}, all honest replicas have the chain that $B$ is extending before entering the new view.
    Therefore, all honest replicas will vote for the proposal.
    For next proposals, the honest leader proposes blocks extending the last proposed block, and all honest replicas will vote for the proposals as well.
    
    Now we show that any honest replica is able to commit $p$ blocks within $6\Delta+(p-1)\alpha$ time. 
    After entering a new view, a time period of
    $6\Delta$ is sufficient for the first block to get committed, since 
    (1) the leader may be at most $\Delta$ later to enter the new view, 
    (2) the leader waits for $2\Delta$ to collect the status message after entering the new view,
    (3) the block takes at most $\Delta$ to reach all honest replicas, which triggers the proposal forwarding,
    (4) any honest replica waits for $\Delta$ before sending \vote,
    (5) the \vote messages take at most $\Delta$ to reach all honest replicas. 
    Then after the first block, there should be one block proposal from the leader in every $\alpha$ time that gets committed in a pipeline fashion. 
    Thus any honest replica should be able to commit $p$ blocks within $6\Delta+(p-1)\alpha$ time. 
    Any honest leader has sufficient time and does not equivocate. 
    Thus, any honest leader will not be blamed by any other honest replica, and a view-change will not occur.
    
    On the other hand, any Byzantine leader will be replaced by a view-change if it sends equivocating blocks or proposals, or does not propose the blocks quickly enough.
    More specifically, if the leader sends equivocating blocks or proposals, then all honest replicas will learn the equivocation and thus send \blame messages to trigger a view-change.
    If the leader does not propose the blocks quickly enough so that all honest replicas send \blame messages, then a view-change is triggered.
    If at least one honest replica is keep committing in time, the \vote certificate broadcasted by this replica when it commits can lead all honest replicas to keep committing new blocks, unless some honest replica gathers a \blame certificate which will lead all honest replicas to perform the view-change.
\end{proof}

\section{Results for Mobile Sluggish Faults}
\label{app:msf}

In this section, we consider the mobile sluggish fault model \cite{chan2018pili, guo2019synchronous, synchotstuff} that 
models the temporary violation of the message delay bound at some honest replicas at a given time. 
The synchronous model assumes that each message sent by honest replicas can reach any honest replica within $\Delta$ time. Such requirement is crucial for the correctness of our protocols \ba, \bb or \smr, where each honest replica expects to receive forwarded proposals from other honest replicas to detect conflicting proposals. 
%
%
In practice, such unforeseen aberrations in the network may happen to any honest replica during the execution of the protocol, especially for BFT SMR protocols that are designed to keep commit values.

We first present the formal definition of mobile sluggish faults \cite{chan2018pili, guo2019synchronous, synchotstuff}, and show how to extend our protocols to tolerate the mobile sluggish faults. 
To illustrate our techniques, we only present the extension for the Byzantine fault tolerant state machine replication protocol \smr.
Similar approaches apply to \ba and \bb as well.

\subsection{The Mobile Sluggish Fault Model}
\label{sec:ms:model}

In the mobile sluggish fault model \cite{chan2018pili, guo2019synchronous, synchotstuff}, an honest replica is either {\em sluggish} or {\em prompt} at a given time. Moreover, the set of sluggish replicas can change over time.
The messages sent or received by sluggish replicas may be delayed in the network, while the prompt replicas follow the message delay bound. 
Formally, if an honest replica $r_1$ is prompt at time $t_1$, then any message sent by $r_1$ at time $\leq t_1$ will be received by any honest replica that is prompt at time $t_2\geq t_1+\Delta$. 
Intuitively, in the mobile sluggish model, any message sent by a sluggish replica would satisfy the maximum message delay bound when the replica becomes prompt.

Following the literature \cite{synchotstuff}, we use $f$ to denote the total number of faults, $d$ to denote the number of honest but sluggish replicas, and $f-d$ to denote the number of Byzantine replicas.
The set of sluggish replicas is mobile in the sense that it can change at any instance of time.
We assume the number of honest and prompt replicas is always $>n/2$ at any time, which is necessary to solve Byzantine agreement or broadcast under the mobile sluggish fault model \cite{guo2019synchronous, synchotstuff}.
Therefore, the total number of replica is $n \geq 2f+1$, and the number of honest and prompt replicas at any time is at least $f+1$.
Without loss of generality, we assume that $n=2f+1$.

\subsection{Protocol \smrms}

In this section, we present a Byzantine fault tolerant state machine replication protocol \smrms under the mobile sluggish model, as in Figure \ref{fig:smr:ms}.
The notations follow from Section~\ref{sec:smr}.
To tolerate mobile sluggish faults, we apply a technique that adds extra communication rounds in the protocol to ensure the message delivery, while achieving the \latency of $\Delta+4\delta$.
More specifically, each Forward/Vote/Blame step in \smr protocol is replaced by two steps in the new protocol, namely Forward, Timer, Vote1, Vote2, Blame1 and Blame2.
The second step (Timer/Vote2/Blame2) can proceed after the replica receives $f+1$ \ackone/\voteone/\blameone messages.

Since now each Forward/Vote/Blame step becomes two steps, we specify which vote or blame message the certificate consists of.
Let $\fc_v(B_k)$ denote the certificate for height-$k$ block $B_k$ in view $v$, consisting of $f+1$ valid \voteone messages for block $B_k$.
Let $\fb_v$ denote the blame certificate in view $v$, consisting of $f+1$ valid \blameone messages.

{\bf Why require $f+1$ \ackone/\voteone/\blameone messages?}
Here we give some intuitive explanations on why the protocol requires $f+1$ \ackone/\voteone/\blameone messages to proceed.
Under the synchronous model, when an honest replica sends an \ack, \vote or \blame message at time $t$, it is guaranteed that all other honest replica can receive the message at time $t+\Delta$.
However, with mobile sluggish faults, the honest replica may be sluggish and the message cannot be received by other replicas. Thus, any protocol that relies on such a single message will fail, such as the protocol \smr that relies on proposal forwarding to detect leader equivocation.
An natural idea is to rely on $\geq f+1$ messages, then at least one of the message is from an honest and prompt replica. 
Therefore, we replace the each Forward, Vote and Blame step in \smr with two communication steps in \smrms, to ensure that an honest replica proceeds only after $f+1$ \ackone/\voteone/\blameone messages are received.

\begin{figure}
    \centering
    \begin{mybox}

\textbf{Steady State Protocol for Replica $i$}

Let $v$ be the current view number and replica $L$ be the leader of the current view. The leader proposes a block every $\alpha$ time, where $\alpha$ is a parameter.
\begin{enumerate}
    \item\label{smrh:propose} {\bf Propose.} The leader $L$ sends $\langle \texttt{propose}, B_k, \fs, v \rangle_L$ to all other replicas, 
    where $B_k=(b_k,h_{k-1})$ is a height-$k$ block, containing a batch of new client requests $b_k$ and a hash digest $h_{k-1}=H(B_{k-1})$ of a height-$(k-1)$ block $B_{k-1}$.
    For the first proposal in a new view after a view-change,
    $\fs$ is a set of $f+1$ distinct signed \status messages received during the view-change, 
    and $B_{k-1}$ is the highest certified block among blocks in $\fs$.
    Otherwise, $\fs=\bot$ and $B_{k-1}$ is the last block proposed by $L$.
    
    \item\label{smrh:forward} {\bf Forward.} Upon receiving a valid new proposal $\langle \texttt{propose}, B_k, v \rangle_L$ for a height-$k$ block, 
    forward the proposal to all other replicas.
    Once one of the following conditions is true: 
    (1) $\fs$ contains $f+1$ distinct signed \status messages, $B_k$ extends the highest certified block in $\fs$, and $B_k$ extends a chain known to the replica, or
    (2) $\fs=\bot$, $B_k$ extends the highest certified block known to the replica, and $B_k$ extends a chain known to the replica, perform the following.
    Broadcast any block in the chain that the replica hasn't forwarded yet, and broadcast an \ackone in the form of $\langle \ackone, B_k, v \rangle_i$.
    
    \item\label{smrh:timer} {\bf Timer.}
    Upon receiving $f+1$ distinct \ackone messages $\langle \ackone, B_k, v \rangle$,
    set $\texttt{vote-timer}_k$ to $\Delta$ and start counting down.
    
    \item\label{smrh:vote1} {\bf Vote1.} When $\texttt{vote-timer}_k$ reaches $0$, if no equivocating blocks/proposals signed by $L$ or a blame certificate $\fb_v$ is received, send a \voteone to all other replicas in the form of $\langle \voteone, B_k, v \rangle_i$.
    
    \item\label{smrh:vote2} {\bf Vote2.} Upon receiving $f+1$ distinct \voteone messages $\langle \voteone, B_k, v \rangle$, form a \voteone certificate $\fc_v(B_k)$.
    If no equivocating blocks/proposals signed by $L$ or a blame certificate $\fb_v$ is received,
    broadcasts \voteone certificate $\fc_v(B_k)$ and 
    a \votetwo in the form of $\langle \votetwo, B_k, v \rangle_i$. 
    
    \item\label{smrh:commit} {\bf Commit.} 
    Upon receiving $f+1$ distinct signed \votetwo messages for block $B_k$, if no equivocating blocks/proposals signed by $L$ or a blame certificate $\fb_v$ are received, broadcast these $f+1$ \votetwo messages, and commit $B_k$ with all its ancestors.
    
\end{enumerate}

\textbf{View-change Protocol for Replica $i$}

Let $L$ and $L'$ be the leader of view $v$ and $v+1$ respectively.
\begin{enumerate}
    \item\label{smrh:blame1} {\bf Blame1.} If less than $p$  valid blocks are committed in $8\Delta+(p-1)\alpha$ time in view $v$, broadcast $\langle \blameone, v\rangle_i$. If equivocating blocks/proposals signed by $L$ are received, broadcast $\langle \blameone, v\rangle_i$ and the two equivocating blocks/proposals.
    
    \item\label{smrh:blame2} {\bf Blame2.} Upon receiving $f+1$ valid \blameone messages $\langle \blameone, v\rangle$, broadcast a \blameone certificate $\fb_v$ of $f+1$ \blameone messages, and a \blametwo in the form of $\langle \blametwo, v \rangle_i$. 
    
    \item\label{smrh:newview} {\bf New-view.}
    Upon gathering $f+1$ distinct $\langle \blametwo, v\rangle$ messages, broadcast these messages, abort all \texttt{vote-timer}(s) and stop sending \ack or \vote in view $v$.
    Wait for $2\Delta$ time and enter view $v+1$.
    Upon entering view $v+1$, send the new leader $L'$ a $\langle \status, B_{k'}, \fc_{v'}(B_{k'}), v \rangle_i$ message where $\fc_{v'}(B_{k'})$ is certificate for the highest certified block $B_{k'}$.
    (Ignore any block certified in view $\leq v$ that is received after this point.)
    If the replica is the new leader $L'$, wait for another $2\Delta$ time upon entering view $v+1$.
    
\end{enumerate}
    \end{mybox}
    \caption{\smrms under the Mobile Sluggish Fault Model}
    \label{fig:smr:ms}
\end{figure}

\stitle{Correctness of Protocol \smrms.}
We first prove two key claims for protocol \smrms, similar to Lemma \ref{lem:smrh:chain} and \ref{lem:smr:safety} for the protocol \smr under synchrony.

\begin{lemma}\label{lem:smrh:chain}
    If a block $B_l$ extending a chain $\fc$ is certified in view $v$, then $f+1$ honest replicas have received all blocks in $\fc$ before entering view $v+1$.
\end{lemma}

\begin{proof}
    Suppose a block $B_l$ extending a chain $\fc$ is certified in view $v$, then a set $A$ of $f+1$ replicas sends \voteone messages for $B_l$.
    Let $t$ denote the earliest time point such that there exist a replica $a\in A$ satisfying the following: (i) $a$ is honest and prompt at time $t$ and (ii) $a$ sends \voteone  at time $\leq t$. 
    The above definition of time $t$ is valid, since the latest time point when an honest replica in $A$ sends \voteone satisfies the above two conditions, which means the candidate set for $t$ is non-empty.
    Since replica $a$ sends \voteone at time $\leq t$, at time $t'\leq t-\Delta$ when setting the $\texttt{vote-timer}_l$ in Step~\ref{smrh:timer}, $a$ receives $f+1$ distinct \ack messages. Since there are $f+1$ honest and prompt replicas at any time, at time $t'$ at least one replica that sends \ack is honest and prompt. Then, this honest and prompt replica already has the chain $\fc$ and forwards all blocks in $\fc$ at time $t'$ according to Step~\ref{smr:forward}.
    Therefore, the set $R$ of $f+1$ honest and prompt replicas at time $t'+\Delta$ receive all blocks in $\fc$. 
    If any replica $r\in R$ enters the new view $v+1$ before receiving the blocks in $\fc$, according to Step~\ref{smr:newview} of the protocol, $r$ must have received $f+1$ distinct \blametwo messages before time $t'-\Delta$ due to the $2\Delta$ waiting period before entering the next view. 
    Therefore, $f+1$ replicas broadcast \blametwo and the \blameone certificate before time $t'-\Delta$. Then, at least one replica $r'$ above is honest and prompt at time $t'-\Delta\leq t-2\Delta$, and has broadcasted the \blameone certificate.
    Consider the honest and prompt replica $a'\in A$ at time $t-\Delta$. By the definition of $t$, replica $a'$ has not set \voteone at time $t-\Delta$. 
    Since replica $a'$ is honest and prompt at time $t-\Delta$, replica $r'$ is honest and prompt at time $t'-\Delta\leq t-2\Delta$ and broadcasts the \blameone certificate, $a'$ receives the \blameone certificate and thus will not send \voteone, causing contradiction.
    Therefore, the set $R$ of $f+1$ honest replicas receive all blocks in $\fc$ before entering view $v+1$.
    
\end{proof}

\begin{lemma}\label{lem:smrh:safety}
    If an honest replica directly commits a block $B_l$ in view $v$, 
    then a certified block that ranks no lower than $\fc_v(B_l)$ must equal or extend $B_l$.
\end{lemma}

\begin{proof}
    Recall that a certified block $B'_{l'}$ with the certificate $\fc_{v'}(B'_{l'})$ ranks no lower than $\fc_{v}(B_l)$ if either (i) $v'=v$ and $l'\geq l$, or (ii) $v'>v$. 
    Suppose that an honest replica $h$ directly commits $B_l$ in view $v$.

    First we prove that any block $B'_{l'}$ with $l'\geq l$ certified in view $v$ must equal or extend $B_l$.
    Suppose for the sake of contradiction that some equivocating block $B_{l'}'$ is certified in view $v$, then there is a set $R'$ of $f+1$ replicas that sends $\voteone'$ for $B_{l'}'$. 
    Since the replica $h$ commits $B_l$, it receives $f+1$ valid \votetwo messages, which implies that at least one honest replica has sent \votetwo after receiving $f+1$ \voteone messages.
    Let the set of $f+1$ replicas that sends \voteone messages above be $R$, and let $t$ be the earliest time point that some honest replica $r$ in $R$ sends \voteone. The definition of $t$ is valid since at least one replica in $R$ is honest.
    Also let $t'$ be the earliest time point that some honest replica $r'$ in $R'$ sends $\voteone'$.
    \begin{itemize}
        \item Suppose that $t'\leq t$. According to the protocol, at time point $t'-\Delta$, the honest replica $r'$ receives $f+1$ $\ackone'$ messages for $B_{l'}'$ and thus $f+1$ replicas have sent $\ackone'$ at $t'-\Delta$.
        Since there are $f+1$ honest and prompt replicas at any time point, at least one of the above replicas that sends $\ackone'$ is honest and prompt at time $t'-\Delta$. According to Step~\ref{smrh:forward}, this replica has also sent all blocks in the chain $\fc'$ that $B_{l'}'$ extends.
        Similarly, at least one of the replicas in $R$ is honest and prompt at time $t'\leq t$, and this replica should receive the $\ackone'$ message and all blocks in the chain $\fc'$.
        Since this replica already has the chain that $B_l$ extends according to Step~\ref{smrh:forward}, it will discover that $B_l$ and $B'_{l'}$ equivocates and thus not sending \voteone. This is a contradiction, hence no equivocating block $B_{l'}'$ can be certified in this case.
        
        \item Suppose that $t'>t$. Similar to the previous case, at time point $t-\Delta$, at least one honest and prompt replica sends $\ackone$. Again, at least one honest and prompt replica in $R'$ should receive the $\ackone$ message and not send $\voteone'$. This contradicts the assumption that $b'$ is certified, and hence no equivocating block $B_{l'}'$ can be certified in this case.
    \end{itemize}
    
    Now we show that any block $B'_{l'}$ certified in view $v'>v$ must equal or extend $B_l$.
    
    We first show that $f+1$ honest replicas receive $\fc_v(B_l)$ before entering the new view $v+1$.
    Since the honest replica $h$ directly commits $B_l$, $h$ receives distinct \votetwo messages from a set $A$ of $f+1$ replicas.
    Let $t$ denote the earliest time point such that there exist a replica $a\in A$ satisfying the following: (i) $a$ is honest and prompt at time $t$ and (ii) $a$ sends \votetwo before or at time $t$. 
    The above definition of time $t$ is valid, since the latest time point when an honest replica in $A$ sends \votetwo satisfies the above two conditions, which means the candidate set for $t$ is non-empty.
    Since replica $a$ is honest and prompt at time $t$ and sends \votetwo before or at time $t$,
    by the definition of the honest and prompt, replica $a$'s \votetwo message and \voteone certificate $\fc_v(B_l)$ will reach the set $R$ of all honest and prompt replica at time $t+\Delta$.
    We will prove that $R$ is the set of $f+1$ honest replicas that receive $\fc_v(B_l)$ before entering the new view $v+1$.
    Suppose for the sake of contradiction, any honest replica $r\in R$ enters the new view $v+1$ at time $t'+\Delta$ where $t'<t$ before receiving $\fc_v(B_l)$.
    According to the protocol, the honest replica $r$ must have received $f+1$ \blametwo messages at time point $t'-\Delta$ due to the $2\Delta$ waiting period before entering the next view. 
    Among the $f+1$ replicas that send \blametwo above, at least one replica $r'$ is honest and prompt at time $t'-\Delta$. According to the protocol, $r'$ receives the \blameone certificate before time $t'-\Delta$ and forwards the \blameone certificate to all other replicas.
    Now consider the honest and prompt replica $a'\in A$ at time $t'$, since $t'<t$, by the definition of $t$, $a'$ has not sent \votetwo. 
    Since $r'$ is honest and prompt at time $t'-\Delta$, $a'$ is honest and prompt at time $t'$,  $a'$ should receive the \blameone certificate $\fb_v$ sent by $r'$ at time $t'$. This will prevent $a'$ from sending the \votetwo message, causing contradiction. 
    Hence, $R$ is the set of $f+1$ honest replicas that receive $\fc_v(B_l)$ before entering the new view $v+1$.
    By a similar proof of Lemma \ref{lem:smrh:chain}, we can also show that $R$ receive the chain that $B_l$ extends.
    
    Together with the claim that any block $B'_{l'}$ certified in view $v$ must equal or extend $B_l$, the highest certified block at any honest replicas equals or extends $B_l$ when entering view $v+1$.
    If the proposal contains $\fs\neq\bot$, then $\fs$ contains at least one \status message from an honest replica, and thus the block in the proposal must extend $B_l$ in order to get certified.
    If the proposal contains $\fs=\bot$, the honest replicas will only vote for blocks that extends a known chain and the highest certified block, and thus, only blocks that extend $B_l$ can be certified.
    Hence, the highest certified block at any honest replicas still equals or extends $B_l$ in view $v+1$.
    By simple induction, in any future view $v'>v$, the highest certified block at any honest replicas must equal or extend $B_l$, therefore only blocks that equal or extending $B_l$ can be certified.
\end{proof}

\begin{theorem}[Safety]\label{thm:smrh:safety}
    Honest replicas always commit the same block $B_k$ for each height $k$.
\end{theorem}
\begin{proof}
    Suppose two blocks $B_k$ and $B_k'$ are committed at height $k$ at two honest replicas. Suppose $B_k$ is committed due to $B_l$ being directly committed in view $v$, and $B_k'$ is committed due to $B_{l'}'$ being directly committed in view $v'$.
    Without loss of generality, suppose $v\leq v'$, and for $v=v'$, further assume that $l\leq l'$. 
    Since $B_l$ is directly committed and $B_{l'}$ is certified and ranks no lower than $\fc_v(B_l)$,
    by Lemma \ref{lem:smrh:safety}, $B_{l'}$ must equal or extend $B_l$. Thus, $B_k'=B_k$.
\end{proof}

\begin{theorem}[Liveness]\label{thm:smrh:liveness}
    All honest replicas keep committing new blocks during the periods in which all honest replicas stay prompt.
\end{theorem}

\begin{proof}
    We assume all honest replicas stay prompt.
    
    If the leader $L$ is honest, a view-change will not occur and all honest replicas keep committing new blocks. 
    The $2\Delta$ waiting window after entering the new view is sufficient for the new leader to collect all status messages from the honest replicas, since any honest replica enters the new view at most $\Delta$ later than the leader and the status message takes at most $\Delta$ time to reach $L$.
    Therefore, the honest leader is able to propose a block extending the highest certified block $B$ among $f+1$ distinct signed \status messages as the first proposal.
    By Lemma \ref{lem:smrh:chain}, $f+1$ honest replicas have the chain that $B$ is extending before entering the new view.
    Therefore, all honest replicas will vote for the proposal.
    For next proposals, the honest leader proposes blocks extending the last proposed block, and all honest replicas will vote for the proposals as well.

    Now we show that any honest replica is able to commit $p$ blocks within $8\Delta+(p-1)\alpha$ time. 
    After entering a new view, a time period of
    $8\Delta$ is sufficient for the first block to get committed, since 
    (1) the leader may be at most $\Delta$ later to enter the new view, 
    (2) the leader waits for $2\Delta$ after entering the new view, 
    (3) the proposal takes at most $\Delta$ to reach all honest replicas, which triggers the \ackone message,
    (4) the \ackone messages take at most $\Delta$ to reach all honest replicas,
    (5) any honest replica waits for $\Delta$ before sending \voteone,
    (6) the \voteone messages take at most $\Delta$ to reach all honest replicas,
    (7) the \votetwo messages take at most $\Delta$ to reach all honest replicas.
    Then after the first block, there should be one block proposal from the leader in every $\alpha$ time that gets committed in a pipeline fashion.
    Therefore, an honest leader has sufficient time and does not equivocate, and thus will not be blamed by any honest replicas. Then, all honest replicas will exchange the \ack and \vote messages and commit.
    
    On the other hand, any Byzantine leader will be replaced by a view-change if it sends equivocating blocks or does not propose the blocks quickly enough, similar to the proof of Theorem \ref{thm:smr:liveness}.
\end{proof}

\begin{theorem}[Good-case Latency]\label{thm:smrh:efficiency}
    During the periods in which all honest replicas stay prompt, if the leader is honest, then the proposed value will be committed in $\Delta+4\delta$ time after being proposed.
\end{theorem}

\begin{proof}
    When the leader is honest, it proposal will take time $\delta$ to reach all replicas. Then, the \ackone messages take $\delta$ to reach all honest replicas. After receiving the \ackone messages, all honest replica will wait for time $\Delta$ before sending the \voteone message. The \voteone messages take $\delta$ to reach all honest replicas, and trigger them to send \votetwo messages. Finally, the above \votetwo messages reach all honest replicas after $\delta$ time, leading to commit at all honest replicas.
    Therefore, the any proposed value will be committed in $\Delta+4\delta$ time.
\end{proof}

\begin{theorem}\label{thm:smr:ms}
\smrms protocol solves  Byzantine  fault tolerant state machine replication in the authenticated setting under the mobile sluggish fault model, with a \latency of  $\Delta+4\delta$ per command and optimal resilience of $f<n/2$.
\end{theorem}
\begin{proof}
    Proved by Theorem \ref{thm:smrh:safety}, \ref{thm:smrh:liveness} and \ref{thm:smrh:efficiency}.
\end{proof}

\end{document}